\documentclass[11pt,a4paper]{article}
\usepackage{dsfont}

\usepackage{CJK}
\usepackage{hyperref}
\usepackage{doi}
\usepackage[charter]{mathdesign}
\usepackage{graphicx}
\usepackage{fullpage}
\usepackage{amsmath}
\usepackage{amsthm}
\usepackage[dvipsnames]{xcolor} \hypersetup{ colorlinks,
  linkcolor={red!50!black}, citecolor={blue!50!black},
  urlcolor={blue!80!black} } \usepackage[inline]{enumitem}
\setlist{noitemsep,leftmargin=*} \newtheorem{thm}{Theorem}
\newtheorem{lem}[thm]{Lemma}

\newcommand{\Xomit}[1]{}

\title{Finding a Shortest Even Hole in Polynomial Time}
\author{Hou-Teng Cheong\footnote{Department of Computer Science and Information Engineering, National Taiwan University.} \and Hsueh-I Lu\footnote{Corresponding author. Department of Computer Science and Information Engineering, National Taiwan University. Email: hil@csie.ntu.edu.tw. Research of this
  author is supported 
  by MOST grant
107--2221--E--002--032--MY3.}}
\date{}

\begin{document} 
\begin{CJK*}{UTF8}{bkai}
\maketitle
\begin{abstract}
An even (respectively, odd) hole in a graph is an induced cycle  with even (respectively, odd) length that is at least four.
Bienstock~[{\em DM}~1991 and 1992] proved that detecting an even (respectively, odd) hole containing a given vertex is NP-complete.
Conforti, Chornu\'ejols, Kappor, and Vu\v{s}kovi\'{c} [{\em FOCS} 1997] gave the first known polynomial-time algorithm to determine whether a graph contains even holes. Chudnovsky, Kawarabayashi, and Seymour~[{\em JGT} 2005] estimated that Conforti et al.'s algorithm runs in $O(n^{40})$ time on an $n$-vertex graph and reduced the required time to $O(n^{31})$.
Subsequently, da~Silva and Vu\v{s}kovi\'{c}~[{\em JCTB}~2013], Chang and Lu~[{\em JCTB}~2017], and Lai, Lu, and Thorup~[{\em STOC}~2020] improved the time to $O(n^{19})$, $O(n^{11})$, and $O(n^9)$, respectively.
The tractability of determining whether a graph contains odd holes has been open for decades until the algorithm of
Chudnovsky, Scott, Seymour, and Spirkl [{\em JACM} 2020]  that runs in $O(n^9)$ time, which Lai et al.~also reduced to $O(n^8)$.
By extending Chudnovsky et al.'s techniques for detecting odd holes, 
Chudnovsky, Scott, and Seymour~[{\em Combinatorica} 2020 to appear] (respectively, [{\em arXiv} 2020]) ensured the tractability of  finding a long (respectively, shortest) odd hole.
They also ensured the NP-hardness of finding a longest odd hole, whose reduction also works for finding a longest even hole.
Recently, Cook and Seymour ensured the tractability of finding a long even hole. An intriguing
missing piece is the tractability of finding a shortest even hole,  left open for at least 15 years by, e.g.,
Chudnovsky et al.~[{\em JGT} 2005] and
Johnson [{\em TALG} 2005].
We resolve this long-standing open problem by giving the first known polynomial-time algorithm, running in $O(n^{31})$ time, for finding a shortest even hole in an $n$-vertex graph that contains even holes.
\end{abstract}

\section{Introduction}
An even (respectively, odd) hole in a graph is an induced cycle with even (respectively, odd) length that is at least four.
Detecting induced subgraphs are 
fundamentally important problems~\cite{ChangL15,ChudnovskyK08,ChudnovskyMSS18,ChudnovskyPPT20,ChudnovskyS10,DalirrooyfardVW19,DiotRTV20,
KaminskiN12,KawarabayashiK12,
Kriesell01a,KwanLST20,
Le19}.
A most prominent example regarding detecting induced cycles is the seminal strong perfect graph theorem of Chudnovsky, Robertson, Seymour, and Thomas~\cite{ChudnovskyRST06}, conjectured by Berge in 
1960~\cite{Berge60,Berge61,Berge85}, ensuring that the perfection of a graph can be determined by detecting odd holes  in the graph and its complement.
Chudnovsky, Cornu{\'{e}}jols, Liu, Seymour, and Vu\v{s}kovi\'{c}~\cite{ChudnovskyCLSV05} gave the first known polynomial-time algorithm, running in $O(n^9)$ time, for recognizing $n$-vertex perfect graphs.
Bienstock~\cite{Bienstock91,Bienstock92} proved that detecting an odd hole containing a prespecified vertex is NP-hard in the early 1990s. The tractability of detecting an odd hole remained unknown until 
the recent $O(n^9)$-time algorithm of
Chudnovsky, Scott, Seymour, and Spirkl~\cite{ChudnovskySSS20}. Lai, Lu, and Thorup~\cite{LaiLT20} implemented Chudnovsky et al.'s algorithm to run in $O(n^8)$ time, also leading to an $O(n^8)$-time algorithm for recognizing perfect graphs. Chudnovsky, Scott, and Seymour~\cite{ChudnovskySS20} showed that it takes $O(n^{20\ell+43})$ time to detect an odd hole with length at least $\ell$.
Chudnovsky, Scott, and Seymour~\cite{ChudnovskySS20-shortest-odd-hole} ensured the NP-hardness of finding a longest odd hole and gave an $O(n^{14})$-time algorithm to obtain a shortest odd hole, if one exists. 

\Xomit{
\begin{thm}
\label{theorem:theorem2}
For any $n$-vertex graph $G$, 
it takes $O(n^{12})$ time to 
either obtain a shortest odd hole of $G$ or 
ensure that $G$ contains no odd hole.
\end{thm}
}

Even holes in graphs have also been extensively
studied in the literature~\cite{Addario-BerryCHRS08,ConfortiCKV00,
  ConfortiCKV02a,daSilvaV07,daSilvaV13,FraserHH18,HusicTT19,KloksMV09,SilvaSS10,WuX19}.
Vu\v{s}kovi\'{c} \cite{Vuskovic10} gave a comprehensive survey on even-hole-free graphs.
According to Bienstock~\cite{Bienstock91,Bienstock92}, detecting an even hole containing a prespecified vertex is also NP-hard. 
Conforti, Cornu{\'e}jols, Kapoor, and
Vu\v{s}kovi\'{c}~\cite{ConfortiCKV97con,ConfortiCKV02b} gave the first
polynomial-time algorithm for detecting even holes in an $n$-vertex graph, running in
$O(n^{40})$ time.  Chudnovsky, Kawarabayashi, and
Seymour~\cite{ChudnovskyKS05} reduced the time to $O(n^{31})$.  
Chudnovsky
et al.~\cite{ChudnovskyKS05} also observed that the time of detecting
even holes can be further reduced to $O(n^{15})$ as long as detecting
prisms is not too expensive, but this turned out to be
NP-hard~\cite{MaffrayT05}.  
Chudnovsky and
Kapadia~\cite{ChudnovskyK08} and Maffray and
Trotignon~\cite[Algorithm~2]{MaffrayT05} devised $O(n^{35})$-time and
$O(n^5)$-time algorithms for detecting prisms in theta-free and
pyramid-free graphs, respectively.  Later, da~Silva and
Vu\v{s}kovi\'{c}~\cite{daSilvaV13} improved the time of detecting even
holes in $G$ to $O(n^{19})$.  Chang and Lu~\cite{ChangL15} further reduced the running time to $O(n^{11})$.  The best currently known algorithm for detecting even holes, due to Lai, Lu, and Thorup~\cite{LaiLT20}, runs in $O(n^9)$ time.
Very recently, Cook and Seymour~\cite{CookS20} announced an $O(n^{54\ell+58})$-time algorithm for detecting even holes with length at least $\ell$.
Following the approach of Chudnovsky et al.~\cite{ChudnovskySS20-shortest-odd-hole} for the NP-hardness of longest odd hole, one can verify that the longest even hole problem is NP-hard by reducing from the problem of determining whether the $n$-vertex graph $G$ admits a Hamiltonian $uw$-path for two given vertices $u$ and $w$: For the graph  $H$ obtained from $G$ by subdividing each edge once and then adding a path $uvw$, a longest even hole of $H$ has $2n$ vertices if and only if $G$ admits a Hamiltonian $uw$-path.
As displayed in Figure~\ref{figure:figure1},
the complexity of finding a shortest even hole, open for at least $15$ years (see, e.g.,~\cite[Page~86]{ChudnovskyKS05} and~\cite[page 166]{Johnson05}), became the only missing piece.
We resolve the long-standing open problem
by presenting the first known polynomial-time algorithm, as summarized in the following theorem.
\begin{thm}
\label{theorem:theorem1}
For any $n$-vertex graph $G$, 
it takes $O(n^{31})$ time to 
either obtain a shortest even hole of $G$ or 
ensure that $G$ contains no even hole.
\end{thm}

\renewcommand{\arraystretch}{1.2}
\begin{figure}
\centerline{
\begin{tabular}{|r|r|r|}
\hline
&odd hole&even hole\\
\hline
containing a vertex&NP-hard~\cite{Bienstock91,Bienstock92}&NP-hard~\cite{Bienstock91,Bienstock92}\\
\hline
longest&NP-hard~\cite{ChudnovskySS20-shortest-odd-hole}&NP-hard~\cite{ChudnovskySS20-shortest-odd-hole}\\
\hline
at least $\ell$ edges
&$O(n^{20\ell+43})$~\cite{ChudnovskySS20}
&$O(n^{54\ell+58})$~\cite{CookS20}\\
\hline
detection&$O(n^8)$~\cite{LaiLT20}&$O(n^9)$~\cite{LaiLT20}\\
\hline
shortest&$O(n^{14})$~\cite{ChudnovskySS20-shortest-odd-hole}&$O(n^{31})$ [this paper]\\
\hline
\end{tabular}
}
\label{figure:figure1}
\caption{The state-of-the-art of hole detection.}
\end{figure}






\Xomit{\begin{thm}
\label{theorem:theorem2}
It is NP-complete to obtain a longest even hole of a graph that contains even holes. \end{thm}
}

Our shortest-even-hole algorithm is based upon Chudnovsky et al.'s techniques~\cite{ChudnovskyKS05}
that lead to their $O(n^{31})$-time algorithm for detecting even holes.
To our surprise, their
 techniques for detecting even holes suffice  to resolve the problem that they left open in the first paragraph of their paper, writing ``the complexity of finding the shortest even hole in a graph is still open as far as we know''.
It is perhaps their more general settings of allowing for weighted graphs that caused them to 
overlook the possibility of further pushing for a shortest even hole, e.g., in their Lemma~5.3.

\Xomit{
The rest of the paper is organized as follows.
Since Theorem~\ref{theorem:theorem1} is the major contribution of the paper, we first prove Theorem~\ref{theorem:theorem1} in 
section~\ref{section:section2}.
Section~\ref{section:section3} proves Theorem~\ref{theorem:theorem2}.
Section~\ref{section:section4} concludes the paper.
}
\Xomit{
The way it constructs an even hole in Theorem 5.3 is to find two paths with different parity and find a path between $c_1, c_i$, which then can construct an even hole with one of them depending on the parity.   But in this case, the shortest path between $c_1, c_i$ in the graph $G \setminus W$ is in fact have a same length and parity as $C(c_i, c_1)$.

For instance,  there can be two shortest paths with different parity between the same end-vertices

It is to our surprise, since
they wrote that
``the complexity of finding the shortest even hole in a graph is still open as far as we know''. Possible explanation is that their settings  allow graphs to be weighted.

\begin{itemize}
\item In Theorem 5.3, since there is weight on graph there could be 2 shortest path between same pair of vertices with different parity. In this case, the phrase 2 exists but not in the unweighted case.
\item The way it constructs an even hole in Theorem 5.3 is to find two paths with different parity and find a path between $c_1, c_i$, which then can construct an even hole with one of them depending on the parity.   But in this case, the shortest path between $c_1, c_i$ in the graph $G \setminus W$ is in fact have a same length and parity as $C(c_i, c_1)$.
\item  The proof of Theorem 6.1 actually implies a shortest even hole in unweighted case.
\end{itemize}
}

\section{Proving Theorem~\ref{theorem:theorem1}}
\label{section:section2}

We first reduce Theorem~\ref{theorem:theorem1} to Lemmas~\ref{lemma:lemma4} and~\ref{lemma:lemma5} via Lemmas~\ref{lemma:lemma1} and~\ref{lemma:lemma3} and then prove Lemmas~\ref{lemma:lemma4} and~\ref{lemma:lemma5} in \S\ref{subsection:subsection2.1} and~\S\ref{subsection:subsection2.2}.

\begin{lem}[{Lai, Lu, and Thorup~\cite[Theorem~1.6]{LaiLT20}}]
\label{lemma:lemma1}
For any $n$-vertex graph $G$, it takes $O(n^9)$ time to determine whether $G$ contains even holes.
\end{lem}

Let $[i,k]$ for integers $i$ and $k$ consist of the integers $j$ with $i\leq j\leq k$.
Let $|S|$ denote the cardinality of set $S$.
Let $\|G\|$ denote the number of edges of graph $G$. Let $V(G)$ consist of the vertices of graph $G$.
For any subgraph $H$ of $G$, let $G[H]$ denote the subgraph of $G$ induced by $V(H)$.
Let $d_G(u,v)$ for vertices $u$ and $v$ of graph $G$ denote the distance of $u$ and $v$ in $G$.
A {\em $uv$-path} is a path with end-vertices $u$ and $v$.

Let $G$ be a graph containing even holes. Let $C$ be a shortest even hole of $G$.
Let $P$ be a $uv$-path of $G$ for distinct and nonadjacent vertices $u$ and $v$ of $C$.
$P$ is {\em $C$-good} 
if the union of $P$ and a  $uv$-path of $C$
remains a shortest even hole of $G$.
$P$ is {\em $C$-bad} if $P$ is not $C$-good.
$P$  is {\em $C$-shallow} 
if
\begin{equation}
\label{equation:eq2}
\|P\|\geq d_C(u,v)-1
\end{equation}
and $G[P\cup C_2]$ for the $uv$-paths $C_1$ and $C_2$ of $C$ with $\|C_1\|\leq \|C_2\|$ is a hole.
$P$
is 
a {\em $C$-shortcut} if
\begin{equation}
\label{equation:eq1}
2\leq 
\|P\|\leq d_C(u,v)\text{\quad and\quad} \|P\|<\frac{\|C\|}{4}.
\end{equation}


%
Observe that 
if $P$ is a $C$-good $C$-shortcut, then $\|P\|=d_C(u,v)$.
$C$ is {\em good} in $G$ if each $C$-shortcut in $G$ is $C$-good.
$C$ is {\em bad} in $G$ if $C$ is not good in $G$.
$P$ is a {\em worst} $C$-shortcut in $G$ if 
$P$ is a $C$-bad $C$-shortcut such that either 
(i) $\|P\|=\|P'\|$ and 
$d_C(u,v)\geq d_C(u',v')$ or 
(ii) 
$\|P\|<\|P'\|$
holds for each $C$-bad $C$-shortcut $u'v'$-path $P'$ in $G$.
%
Observe that $C$ is bad in $G$  if and only if there is a worst $C$-shortcut in $G$.
%
A graph $G$ is {\em shallow} if
there is a shortest even hole $C$ of $G$ such that each worst $C$-shortcut is $C$-shallow.
%
A graph $G$ is {\em anti-shallow} if each worst $C$-shortcut for each bad shortest even hole $C$ of $G$ is not $C$-shallow.
Observe that if $G$ is shallow and anti-shallow, then $G$ 
contains a good shortest even hole.

\begin{lem}[{Chudnovsky et al.~\cite[Lemma~4.5]{ChudnovskyKS05}}]
\label{lemma:lemma3}
  For any $n$-vertex graph $G$ that contains even holes, it takes $O(n^{25})$ time to
obtain induced subgraphs
$G_1,\ldots,G_r$ of $G$   with $r=O(n^{23})$ such 
that a $G_i$ with $i\in[1,r]$ is shallow and contains a shortest even hole of $G$.
\Xomit{
 $G[V_i]$ with $i\in[1,r]$ contains a shortest even hole $C$ of $G$ such that either
$C$ is good in $G[V_i]$ or 
each
worst $C$-shortcut in $G[V_i]$ is $C$-shallow.
}
\Xomit{
such that there is an $i\in[1,r]$ satisfying that
\begin{itemize}
\item $G[V_i]$ contains a shortest even hole of $G$ and

\item {\color{NavyBlue}{\color{red}each} worst $C$-shortcut in $G[V_i]$ for {\color{red}each} shortest even hole $C$ of $G[V_i]$ is shallow in $G[V_i]$}.\footnote{\color{red}或放寬成
such that if
$C$ is bad in $G[V_i]$, then some worst $C$-shortcut is shallow.}
\end{itemize}
}
\end{lem}


\begin{lem}
\label{lemma:lemma4}
For any $n$-vertex graph $G$,
it takes $O(n^6)$ time to obtain a subgraph $C$ of $G$ such that (i) $C$ is a shortest even hole of $G$ or 
(ii) $G$ is anti-shallow.
\end{lem}

A graph $G$ is {\em long} if $G$ does not contain any
even hole with 
at most $22$ vertices.
A graph $G$ is {\em bad} if $G$ does not contain any good shortest even hole.

\begin{lem}
\label{lemma:lemma5}
For any $n$-vertex long graph $G$, it takes $O(n^8)$ time to obtain a subgraph $C$ of $G$ such that (i) $C$ is a shortest even hole of $G$ or (ii)
$G$ is bad.
\end{lem}

\Xomit{
\begin{proof}[Proof of Theorem~\ref{theorem:theorem1}]
By Lemma~\ref{lemma:lemma1}, we may assume that $G$ contains even holes. Spend $O(n^{14})$ time to either obtain a shortest even hole $C$ of $G$ or ensure that $G$ has no even holes with at most $16$ vertices. 
For the latter case, 
apply Lemma~\ref{lemma:lemma5} to either obtain a shortest even hole of $G$ or ensure that each shortest even hole of $G$ is bad in $G$.
\end{proof}
}


\begin{proof}[Proof of Theorem~\ref{theorem:theorem1}]
By Lemma~\ref{lemma:lemma1}, we may assume that $G$ contains even holes. 
Spend $O(n^{22})$ time to either obtain a shortest even hole of $G$ or ensure that $G$ is 
long.
%
If $G$ is long, then  
apply Lemma~\ref{lemma:lemma3} in $O(n^{25})$ time to obtain induced subgraphs  $G_1,\ldots,G_r$ of $G$ with $r=O(n^{23})$ such that  $G_i$ for an (unknown) index $i\in[1,r]$ is shallow and contains a(n unknown) shortest even hole $C$ of $G$. 
Since $G$ is long, so is each $G_i$.
For each $j\in[1,r]$, apply
Lemmas
\ref{lemma:lemma4} and~\ref{lemma:lemma5}  on $G_j$ in overall 
$(O(n^6)+O(n^8))\cdot O(n^{23})=O(n^{31})$ time to obtain subgraphs $C_j$ and $D_j$ of $G[V_j]$ such that
 $C_j$ or $D_j$
is a shortest even hole of $G_j$ unless
$G_j$ is anti-shallow and bad.
Finally, we report a $C_j$ or $D_j$ that is an even hole of $G_j$ whose number of edges is minimized over all $j\in[1,r]$. 
Since $G_i$ is shallow,
$G_i$ cannot be  anti-shallow and bad.
Thus, at least one of $C_i$ and $D_i$ is a shortest even hole of $G_i$, which has to be a shortest even hole of $G$ by $C\subseteq G_i$.
\end{proof}
It remains to prove Lemmas~\ref{lemma:lemma4} and~\ref{lemma:lemma5}.

\subsection{Proving Lemma~\ref{lemma:lemma4}}
\label{subsection:subsection2.1}
For any vertex subset $U$ of a graph $G$, let $G-U=G[V(G)\setminus U]$.  
\begin{lem}[{Chudnovsky et al.~\cite[Lemma~5.1]{ChudnovskyKS05}}]
\label{lemma:lemma6}
Let $C$ be a shortest even hole of $G$. 
Let $u$ and $v$ be distinct vertices of $C$.
Let $uv$-path $P$ of $G$ be a $C$-shallow worst $C$-shortcut.
Let $C_1$ and $C_2$ be the $uv$-paths of $C$ with 
$\|C_1\| < \|C_2\|$.
If 
$x$ (respectively, $y$) is the neighbor of $u$ (respectively, $v$) in $C_1$ and
$C_3$ is the 
 $xy$-path of $C_1$, then
the following statements hold: \begin{enumerate}
\item
\label{lemma6:item1}
 If $P_{uv}$ is a $uv$-path of $G$, then 
$\|P\|\leq \|P_{uv}\|$.
\item 
\label{lemma6:item2}
If $P_{uv}$ is a $uv$-path of $G$ with $\|P\|=\|P_{uv}\|$, then 
$G[P_{uv}\cup C_2]$ is a hole of $G$.
\item 
\label{lemma6:item3}
If $P_{xy}$ is an $xy$-path of $G$, then 
$\|C_3\|\leq \|P_{xy}\|$.
\item 
\label{lemma6:item4}
If $P_{xy}$ is an $xy$-path of $G$ with $\|C_3\|=\|P_{xy}\|$, then 
$G[P_{xy}\cup C_2]$ is a hole of $G$. \end{enumerate}
\end{lem}

\begin{proof}[Proof of Lemma~\ref{lemma:lemma4}]
Let $G$ be connected without loss of generality.
For any $U\subseteq V(G)$, 
let $N_G[U]$ denote the vertex subset of $G$ consisting of the vertices in $U$ and their neighbors in $G$.
For any vertices $u$ and $v$, let $P_{uv}$ be an arbitrary shortest $uv$-path of $G$.
For any vertex-disjoint
edges $ux$ and $vy$ of $G$ 
with 
$\|P_{uv}\|=\|P_{xy}\|-1$ such that
$u$ and $v$ 
are connected in 
$$H(u,v,x,y)=G[(V(G)\setminus N_G[V(P_{uv}\cup P_{xy})\setminus \{u,v\}])\cup \{u,v\}],$$
let $Q(u,v,x,y)$ be a shortest $uv$-path of $H(u,v,x,y)$.
If  there are edges $ux$ and $vy$ of $G$ such that $Q(u,v,x,y)$ exists and $G[P_{xy}\cup Q(u,v,x,y)]$ is an even hole, then 
report a shortest such even hole;
Otherwise, just report the empty graph. The procedure takes overall $O(n^6)$ time.

The rest of the proof shows that 
the subgraph reported by the above procedure is a shortest even hole of $G$
as long as $G$ is not anti-shallow, i.e., 
a $uv$-path $P$ is a $C$-shallow
worst $C$-shortcut for a bad
shortest even hole $C$ of $G$.
Let $C_1$ and $C_2$ be the $uv$-paths of $C$ with $\|C_1\|\leq \|C_2\|$.
(a) We first show that $G[P_{uv}\cup C_2]$ is a hole of $G$.
By Lemma~\ref{lemma:lemma6}\eqref{lemma6:item2}, 
$G[P\cup C_2]$ is a hole.
By Equations~\eqref{equation:eq2} and \eqref{equation:eq1}, $\|P\|\leq \|C_1\|\leq \|P\|+1$. We have
\begin{equation}
\label{equation:eq3}
\|P\|=\|C_1\|-1
\end{equation}
or else $\|P\|=\|C_1\|$ would imply that 
$G[P\cup C_2]$ is a shortest even hole of $G$,
contradicting that $P$ is a $C$-bad $C$-shortcut. 
By
Lemma~\ref{lemma:lemma6}\eqref{lemma6:item1} and the definition of $P_{uv}$,
\begin{equation}
\label{equation:eq4}
\|P_{uv}\|=\|P\|,
\end{equation}
implying that $G[P_{uv}\cup C_2]$ is a hole of $G$ by Lemma~\ref{lemma:lemma6}\eqref{lemma6:item2}.
(b) Let $x$ (respectively, $y$) be the neighbor of $u$ (respectively, $v$) in $C_1$.
We next show that $G[P_{xy}\cup C_2]$ is a shortest even hole of $G$.
By Lemma~\ref{lemma:lemma6}\eqref{lemma6:item3},
\begin{equation}
\label{equation:eq5}
\|P_{xy}\|=\|C_3\|=\|C_1\|-2.
\end{equation}
By Lemma~\ref{lemma:lemma6}\eqref{lemma6:item4},
$G[P_{xy}\cup C_2]$ is a shortest even hole of $G$.
Since $G[P_{uv}\cup C_2]$ and $G[P_{xy}\cup C_2]$ are holes of $G$, 
the interior of $C_2$ is disjoint from 
$N_G[V(P_{uv}\cup P_{xy}) \setminus \{u,v\}]$, implying  $C_2\subseteq H(u,v,x,y)$ and that  
$u$ and $v$ are connected in $H(u,v,x,y)$.
Therefore, $Q(u,v,x,y)$ exists with 
\begin{equation}
\label{equation:eq6}
\|Q(u,v,x,y)\|\leq \|C_2\|.
\end{equation}
By Equations~\eqref{equation:eq3},~\eqref{equation:eq4}, and~\eqref{equation:eq5}, we have 
$\|P_{uv}\|=\|P_{xy}\|-1$, 
implying that exactly one of $G[P_{uv}\cup Q(u,v,x,y)]$ and
$G[P_{xy}\cup Q(u,v,x,y)]$ is an even hole of $G$. 
By Equations~\eqref{equation:eq3},~\eqref{equation:eq4},~\eqref{equation:eq5},~and~\eqref{equation:eq6},
\begin{eqnarray*}
\|G[P_{uv}\cup Q(u,v,x,y)]\|&\leq&\|C\|-1\\
\|G[P_{xy}\cup Q(u,v,x,y)]\|&\leq&\|C\|.
\end{eqnarray*}
Therefore, $G[P_{xy}\cup Q(u,v,x,y)]$ is a shortest even hole of $G$.
\end{proof}
\subsection{Proving Lemma~\ref{lemma:lemma5}}

\begin{proof}[Proof of Lemma~\ref{lemma:lemma5}]
\label{subsection:subsection2.2}

Let the long  graph $G$ be connected without loss of generality. For each  $i\in [0,7]$, let $i^+=(i+1)\bmod 8$.
For each of the $O(n^8)$ choices of distinct vertices $v_0,\ldots,v_7$, let $$C(v_0,\ldots,v_7)=P_0\cup\cdots\cup P_7,$$ where each $P_i$ with $i\in[0,7]$ is an arbitrary shortest $v_iv_{i^+}$-path of $G$.
If one of these $O(n^8)$ subgraphs $C(v_0,\ldots,v_7)$ is an even hole of $G$ with $\|P_i\|\geq 3$ for each $i\in[0.7]$, then  report  a shortest such one; Otherwise, just report the empty graph.
A naive implementation of the algorithm takes $O(n^{10})$ time.
The algorithm can  be  implemented to run in $O(n^8)$ time:
Spend overall  $O(n^4)$ time to obtain $d_G(u,v)$ and a shortest $uv$-path $P(u,v)$ of $G$ for any 
vertices $u$ and $v$. 
Spend
overall $O(n^6)$ time to obtain a data structure from which it takes $O(1)$ time to determine (1) whether $G[P(u,v)\cup P(v,w)]$ is a path for any three vertices $u$, $v$, and $w$ and (2) whether
$G[P(u,v)\cup P(x,y)]$ is disconnected for any four vertices $u$, $v$, $x$, and $y$.
It then takes $O(1)$ time for any given vertices $v_0,\ldots,v_7$ to obtain $\|C(v_0,\ldots,v_7)\|$ and whether
$C(v_0,\ldots,v_7)$ is an even hole with $\|P_i\|\geq 3$ for each $i\in[0,7]$.

For the correctness, the rest of the proof shows that 
if $G$ is not bad, i.e., there is a good shortest even hole $C$ of $G$, then 
one of the $O(n^8)$ iterations yields a shortest even hole of $G$.
Since $G$ is long, $\|C\|\geq 24$, implying integers $a\geq 3$ and $b\in [0,7]$ with 
$$\|C\| = 8a + b.$$
Let $v_0,\ldots,v_7$ be vertices of $C$ such that
the  shortest $v_iv_{i^+}$-paths $C_i$ of $C$  with $i\in [0,7]$ are edge-disjoint and satisfy 
$$\|C_i\|\in\{a,a+1\} \text{\quad and\quad} 
\|C_i\|+\|C_{i^+}\|\leq 2a+\left\lceil \frac{b}{4}\right\rceil
.
$$
1. We first ensure for each $i\in[0,7]$ that
 $$\|P_i\|=\|C_i\|.$$
By $d_C(v_i,v_{i^+})\geq a\geq 3$ and the fact that $C$ is a hole, we have $\|P_i\|\geq 2$.
By
$\|P_i\|=d_G(v_i,v_{i^+})\leq d_C(v_i,v_{i^+})=\|C_i\|\leq  a+1<2a\leq \frac{\|C\|}{4}$, 
$P_i$ is a $C$-shortcut. Since $C$ is a good shortest even hole of $G$, the $C$-shortcut $P_i$ is $C$-good, implying
$\|P_i\|=d_C(v_i,v_{i^+})=\|C_i\|.$

2.
Assume for contradiction that a $G[P_i\cup P_{i^+}]$ with $i\in[0,7]$ is not a path.
Thus, 
$\|P\|<\|C_i\|+\|C_{i^+}\|$
holds for a
shortest $v_iv_{i^{++}}$-path $P$ of $G[P_i\cup P_{i^+}]$.
Since $C$ is good and the union of $P$ and the longer $v_iv_{i^{++}}$-path of $C$ is a hole of $G$, 
$P$ cannot be a $C$-shortcut.
Hence, 
$$
2a+\frac{b}{4}=\frac{\|C\|}{4}\leq \|P\|<\|C_i\|+\|C_{i^+}\|\leq 2a+\left\lceil \frac{b}{4}\right\rceil,$$
contradicting that $\|P\|$ is an integer. Therefore, each $G[P_i\cup P_{i^+}]$ with $i\in[0,7]$ is a path.

3. To see that $G[P_0\cup \cdots\cup P_7]$ is a hole, 
assume for contradiction  integers $i\in[0,7]$ and  $d\in[1,3]$ such that
$G[P_i\cup P_j]$
with $j=(i^++d)\bmod 8$ is connected.
There are $v_iv_j$-path $Q$ and $v_{i^+}v_{j^+}$-path $R$ of $G[P_i\cup P_j]$ with $|V(Q)\cap V(R)|\leq 2$.
Thus, 
\begin{equation}
\label{equation:eq7}
\|Q\|+\|R\|\leq \|P_i\|+\|P_j\|+2=\|C_i\|+\|C_j\|+2\leq 2a+4.
\end{equation}
We have
\begin{equation}
\label{equation:eq8}
\|Q\|\geq a+3
\end{equation}
 or
else
$
\|Q\|\leq a+2<2a\leq\min\left\{d_C(v_i,v_j),\frac{\|C\|}{4}\right\}$
and the fact that $C$ is good would imply that $Q$ is a $C$-good $C$-shortcut, 
contradicting  $\|Q\|<d_C(v_i,v_j)$.
Similarly, we have
\begin{equation}
\label{equation:eq9}
\|R\|\geq a+3
\end{equation}
 or else
$
\|R\|\leq a+2<2a\leq\min\left\{
d_C(v_{i^+},v_{j^+}),\frac{\|C\|}{4}
\right\}
$
and the fact that $C$ is good would imply that $R$ is a $C$-good $C$-shortcut, contradicting $\|R\|<d_C(v_{i^+},v_{j^+})$.
Combining Equations~\eqref{equation:eq7}, \eqref{equation:eq8}, and~\eqref{equation:eq9}, 
we have
$2a+6\leq \|Q\|+\|R\|\leq 2a+4$,
contradiction.
\end{proof}


\section{Concluding remarks}
\label{section:section4}
We
resolve the long-standing open problem on the tractability of   reporting a shortest even hole in an $n$-vertex graph by presenting an
 $O(n^{31})$-time algorithm. The complexity is much higher than that of 
 the current $O(n^{10})$ time for reporting an arbitrary even hole, implied by the $O(n^9)$-time algorithm of Lai~et al.~\cite[Theorem~1.6]{LaiLT20} for detecting even holes.
The current time for reporting an arbitrary odd hole is $O(n^9)$, implied by the $O(n^8)$-time algorithm of Lai et al.~\cite[Theorem~1.4]{LaiLT20} for detecting odd holes.
The shortest-odd-hole algorithm of
Chudnovsky et al.~\cite{ChudnovskySS20-shortest-odd-hole} runs in $O(n^{14})$ time.
The $O(n^5)$ gap for odd holes is much smaller than the $O(n^{21})$ gap  for even holes. It would be of interest
to reduce either one of the gaps.

\bibliographystyle{abbrv}
\bibliography{seh}

\begin{thebibliography}{10}

\bibitem{Addario-BerryCHRS08}
L.~Addario-Berry, M.~Chudnovsky, F.~Havet, B.~Reed, and P.~Seymour.
\newblock Bisimplicial vertices in even-hole-free graphs.
\newblock {\em Journal of Combinatorial Theory, Series B}, 98(6):1119--1164,
  2008.
\newblock see~\cite{Addario-BerryCH20} for corrigendum,
  \doi{10.1016/j.jctb.2007.12.006}.

\bibitem{Addario-BerryCH20}
L.~Addario{-}Berry, M.~Chudnovsky, F.~Havet, B.~A. Reed, and P.~D. Seymour.
\newblock Corrigendum to ``{B}isimplicial vertices in even-hole-free graphs''.
\newblock {\em Journal of Combinatorial Theory, Series B}, 142:374--375, 2020.
\newblock \doi{10.1016/j.jctb.2020.02.001}.

\bibitem{Berge60}
C.~Berge.
\newblock Les probl\`{e}mes de coloration en th\'{e}orie des graphes.
\newblock {\em Publications de l'Institut de statistique de l'Universit\'{e} de
  Paris}, 9:123--160, 1960.

\bibitem{Berge61}
C.~Berge.
\newblock F\"arbung von {Graphen} deren s\"amtliche bzw.~deren ungerade
  {Kreise} starr sind ({Zusammenfassung}).
\newblock {\em Wissenschaftliche Zeitschrift, Martin Luther Universit\"at
  Halle-Wittenberg, Mathematisch-Naturwissenschaftliche Reihe}, 10:114--115,
  1961.

\bibitem{Berge85}
C.~Berge.
\newblock {\em Graphs}.
\newblock North-Holland, Amsterdam, New York, 1985.

\bibitem{Bienstock91}
D.~Bienstock.
\newblock On the complexity of testing for odd holes and induced odd paths.
\newblock {\em Discrete Mathematics}, 90(1):85--92, 1991.
\newblock See~\cite{Bienstock92} for corrigendum,
  \doi{10.1016/0012-365X(91)90098-M}.

\bibitem{Bienstock92}
D.~Bienstock.
\newblock Corrigendum to: D. {B}ienstock, ``{O}n the complexity of testing for
  odd holes and induced odd paths'' {D}iscrete {M}athematics 90 (1991) 85--92.
\newblock {\em Discrete Mathematics}, 102(1):109, 1992.
\newblock \doi{10.1016/0012-365X(92)90357-L}.

\bibitem{ChangL15}
H.-C. Chang and H.-I. Lu.
\newblock A faster algorithm to recognize even-hole-free graphs.
\newblock {\em Journal of Combinatorial Theory, Series B}, 113:141--161, 2015.
\newblock \doi{10.1016/j.jctb.2015.02.001}.

\bibitem{ChudnovskyCLSV05}
M.~Chudnovsky, G.~Cornu{\'{e}}jols, X.~Liu, P.~D. Seymour, and
  K.~Vu\v{s}kovi\'{c}.
\newblock Recognizing {Berge} graphs.
\newblock {\em Combinatorica}, 25(2):143--186, 2005.
\newblock \doi{10.1007/s00493-005-0012-8}.

\bibitem{ChudnovskyK08}
M.~Chudnovsky and R.~Kapadia.
\newblock Detecting a theta or a prism.
\newblock {\em SIAM Journal on Discrete Mathematics}, 22(3):1164--1186, 2008.
\newblock \doi{10.1137/060672613}.

\bibitem{ChudnovskyKS05}
M.~Chudnovsky, K.-i. Kawarabayashi, and P.~Seymour.
\newblock Detecting even holes.
\newblock {\em Journal of Graph Theory}, 48(2):85--111, 2005.
\newblock \doi{10.1002/jgt.20040}.

\bibitem{ChudnovskyMSS18a}
M.~Chudnovsky, F.~Maffray, P.~D. Seymour, and S.~Spirkl.
\newblock Corrigendum to ``{E}ven pairs and prism corners in square-free
  {B}erge graphs'' {[J. Combin. Theory, Ser. B 131 (2018) 12--39]}.
\newblock {\em Journal of Combinatorial Theory, Series {B}}, 133:259--260,
  2018.
\newblock \doi{10.1016/j.jctb.2018.07.004}.

\bibitem{ChudnovskyMSS18}
M.~Chudnovsky, F.~Maffray, P.~D. Seymour, and S.~Spirkl.
\newblock Even pairs and prism corners in square-free {Berge} graphs.
\newblock {\em Journal of Combinatorial Theory, Series {B}}, 131:12--39, 2018.
\newblock See~\cite{ChudnovskyMSS18a} for corrigendum,
  \doi{10.1016/j.jctb.2018.01.003}.

\bibitem{ChudnovskyPPT20}
M.~Chudnovsky, M.~Pilipczuk, M.~Pilipczuk, and S.~Thomass{\'{e}}.
\newblock On the maximum weight independent set problem in graphs without
  induced cycles of length at least five.
\newblock {\em SIAM Journal on Discrete Mathematics}, 34(2):1472--1483, 2020.
\newblock \doi{10.1137/19M1249473}.

\bibitem{ChudnovskyRST06}
M.~Chudnovsky, N.~Robertson, P.~Seymour, and R.~Thomas.
\newblock The strong perfect graph theorem.
\newblock {\em Annals of Mathematics}, 164(1):51--229, 2006.
\newblock \doi{10.4007/annals.2006.164.51}.

\bibitem{ChudnovskySS20-shortest-odd-hole}
M.~Chudnovsky, A.~Scott, and P.~Seymour.
\newblock Finding a shortest odd hole.
\newblock {\em arXiv}, 2020.
\newblock \url{https://arxiv.org/abs/2004.11874}.

\bibitem{ChudnovskySS20}
M.~Chudnovsky, A.~Scott, and P.~Seymour.
\newblock Detecting a long odd hole.
\newblock {\em Combinatorica}, 2020, to appear.
\newblock \url{https://arxiv.org/abs/1904.12273}.

\bibitem{ChudnovskySSS20}
M.~Chudnovsky, A.~Scott, P.~Seymour, and S.~Spirkl.
\newblock Detecting an odd hole.
\newblock {\em Journal of the ACM}, 67(1):5:1--5:12, 2020.
\newblock \doi{10.1145/3375720}.

\bibitem{ChudnovskyS10}
M.~Chudnovsky and P.~Seymour.
\newblock The three-in-a-tree problem.
\newblock {\em Combinatorica}, 30(4):387--417, 2010.
\newblock \doi{10.1007/s00493-010-2334-4}.

\bibitem{ConfortiCKV97con}
M.~Conforti, G.~Cornu{\'e}jols, A.~Kapoor, and K.~Vu\v{s}kovi\'{c}.
\newblock Finding an even hole in a graph.
\newblock In {\em Proceedings of the 38th Symposium on Foundations of Computer
  Science}, pages 480--485, 1997.
\newblock \doi{10.1109/SFCS.1997.646136}.

\bibitem{ConfortiCKV00}
M.~Conforti, G.~Cornu{\'e}jols, A.~Kapoor, and K.~Vu\v{s}kovi\'{c}.
\newblock Triangle-free graphs that are signable without even holes.
\newblock {\em Journal of Graph Theory}, 34(3):204--220, 2000.
\newblock \doi{10.1002/1097-0118(200007)34:3<204::AID-JGT2>3.0.CO;2-P}.

\bibitem{ConfortiCKV02a}
M.~Conforti, G.~Cornu{\'e}jols, A.~Kapoor, and K.~Vu\v{s}kovi\'{c}.
\newblock Even-hole-free graphs {Part~I}: Decomposition theorem.
\newblock {\em Journal of Graph Theory}, 39(1):6--49, 2002.
\newblock \doi{10.1002/jgt.10006}.

\bibitem{ConfortiCKV02b}
M.~Conforti, G.~Cornu{\'e}jols, A.~Kapoor, and K.~Vu\v{s}kovi\'{c}.
\newblock Even-hole-free graphs {Part~II}: Recognition algorithm.
\newblock {\em Journal of Graph Theory}, 40(4):238--266, 2002.
\newblock \doi{10.1002/jgt.10045}.

\bibitem{CookS20}
L.~Cook and P.~Seymour.
\newblock Detecting a long even hole.
\newblock
  \url{https://web.math.princeton.edu/~pds/papers/longevenholes/paper.pdf}.

\bibitem{daSilvaV07}
M.~V.~G. da~Silva and K.~Vu\v{s}kovi\'{c}.
\newblock Triangulated neighborhoods in even-hole-free graphs.
\newblock {\em Discrete Mathematics}, 307(9-10):1065--1073, 2007.
\newblock \doi{10.1016/j.disc.2006.07.027}.

\bibitem{daSilvaV13}
M.~V.~G. da~Silva and K.~Vu\v{s}kovi\'{c}.
\newblock Decomposition of even-hole-free graphs with star cutsets and
  $2$-joins.
\newblock {\em Journal of Combinatorial Theory, Series B}, 103(1):144--183,
  2013.
\newblock \doi{10.1016/j.jctb.2012.10.001}.

\bibitem{DalirrooyfardVW19}
M.~Dalirrooyfard, T.~D. Vuong, and V.~Vassilevska~Williams.
\newblock Graph pattern detection: hardness for all induced patterns and faster
  non-induced cycles.
\newblock In {\em Proceedings of the 51st Symposium on Theory of Computing},
  pages 1167--1178, 2019.
\newblock \doi{10.1145/3313276.3316329}.

\bibitem{DiotRTV20}
E.~Diot, M.~Radovanovi{\'c}, N.~Trotignon, and K.~Vu{\v s}kovi{\'c}.
\newblock The (theta, wheel)-free graphs {Part~I}: Only-prism and only-pyramid
  graphs.
\newblock {\em Journal of Combinatorial Theory, Series B}, 143:123--147, 2020.
\newblock \doi{10.1016/j.jctb.2017.12.004}.

\bibitem{FraserHH18}
D.~J. Fraser, A.~M. Hamel, and C.~T. Ho{\`{a}}ng.
\newblock On the structure of (even hole, kite)-free graphs.
\newblock {\em Graphs and Combinatorics}, 34(5):989--999, 2018.
\newblock \doi{10.1007/s00373-018-1925-5}.

\bibitem{HusicTT19}
E.~Husic, S.~Thomass{\'{e}}, and N.~Trotignon.
\newblock The independent set problem is {FPT} for even-hole-free graphs.
\newblock In {\em Proceedings of the 14th International Symposium on
  Parameterized and Exact Computation}, pages 21:1--21:12, 2019.
\newblock \doi{10.4230/LIPIcs.IPEC.2019.21}.

\bibitem{Johnson05}
D.~S. Johnson.
\newblock The {NP}-completeness column.
\newblock {\em ACM Transactions on Algorithms}, 1(1):160--176, 2005.
\newblock \doi{10.1145/1077464.1077476}.

\bibitem{KaminskiN12}
M.~Kaminski and N.~Nishimura.
\newblock Finding an induced path of given parity in planar graphs in
  polynomial time.
\newblock In {\em Proceedings of the 23rd Symposium on Discrete Algorithms},
  pages 656--670, 2012.
\newblock \doi{10.1137/1.9781611973099.55}.

\bibitem{KawarabayashiK12}
K.~Kawarabayashi and Y.~Kobayashi.
\newblock A linear time algorithm for the induced disjoint paths problem in
  planar graphs.
\newblock {\em Journal of Computer and System Sciences}, 78(2):670--680, 2012.
\newblock \doi{10.1016/j.jcss.2011.10.004}.

\bibitem{KloksMV09}
T.~Kloks, H.~M\"{u}ller, and K.~Vu\v{s}kovi\'{c}.
\newblock Even-hole-free graphs that do not contain diamonds: A structure
  theorem and its consequences.
\newblock {\em Journal of Combinatorial Theory, Series B}, 99(5):733--800,
  2009.
\newblock \doi{10.1016/j.jctb.2008.12.005}.

\bibitem{Kriesell01a}
M.~Kriesell.
\newblock Induced paths in $5$-connected graphs.
\newblock {\em Journal of Graph Theory}, 36(1):52--58, 2001.
\newblock \doi{10.1002/1097-0118(200101)36:1<52::AID-JGT5>3.0.CO;2-N}.

\bibitem{KwanLST20}
M.~Kwan, S.~Letzter, B.~Sudakov, and T.~Tran.
\newblock Dense induced bipartite subgraphs in triangle-free graphs.
\newblock {\em Combinatorica}, 40(1):283--305, 2020.
\newblock \doi{10.1007/s00493-019-4086-0}.

\bibitem{LaiLT20}
K.-Y. Lai, H.-I. Lu, and M.~Thorup.
\newblock Three-in-a-tree in near linear time.
\newblock In {\em Proccedings of the 52nd Annual {ACM} Symposium on Theory of
  Computing}, pages 1279--1292, 2020.
\newblock \doi{10.1145/3357713.3384235}.

\bibitem{Le19}
N.~Le.
\newblock Detecting an induced subdivision of {$K_4$}.
\newblock {\em Journal of Graph Theory}, 90(2):160--171, 2019.
\newblock \doi{10.1002/jgt.22374}.

\bibitem{MaffrayT05}
F.~Maffray and N.~Trotignon.
\newblock Algorithms for perfectly contractile graphs.
\newblock {\em SIAM Journal on Discrete Mathematics}, 19(3):553--574, 2005.
\newblock \doi{10.1137/S0895480104442522}.

\bibitem{SilvaSS10}
A.~Silva, A.~A. da~Silva, and C.~L. Sales.
\newblock A bound on the treewidth of planar even-hole-free graphs.
\newblock {\em Discrete Applied Mathematics}, 158(12):1229--1239, 2010.
\newblock \doi{10.1016/j.dam.2009.07.010}.

\bibitem{Vuskovic10}
K.~Vu\v{s}kovi\'{c}.
\newblock Even-hole-free graphs: A survey.
\newblock {\em Applicable Analysis and Discrete Mathematics}, 4(2):219--240,
  2010.
\newblock \doi{10.2298/AADM100812027V}.

\bibitem{WuX19}
R.~Wu and B.~Xu.
\newblock A note on chromatic number of (cap, even hole)-free graphs.
\newblock {\em Discrete Mathematics}, 342(3):898--903, 2019.
\newblock \doi{10.1016/j.disc.2018.11.005}.

\end{thebibliography}
\end{CJK*}
\end{document}